\documentclass[lettersize,onecolumn]{IEEEtran}
\usepackage{amsmath,amsfonts}
\usepackage{algorithmic}
\usepackage{algorithm}
\usepackage{array}
\usepackage[caption=false,font=normalsize,labelfont=sf,textfont=sf]{subfig}
\usepackage{textcomp}
\usepackage{stfloats}
\usepackage{url}
\usepackage{verbatim}
\usepackage{graphicx}
\usepackage{cite}
\usepackage{newtxtext}
\usepackage{amssymb, amsmath,threeparttable, bbding, url,xcolor}
\usepackage{booktabs}
\usepackage{circuitikz}
\usepackage{subcaption}
\usepackage{ulem}

\newtheorem{theorem}{Theorem}[section]
\newtheorem{lemma}[theorem]{Lemma}

\newtheorem{example}[theorem]{Example}

\newtheorem{corollary}[theorem]{Corollary}

\newtheorem{definition}[theorem]{Definition}

\ctikzset{logic ports=ieee}

\newcommand{\1}{\mathbf{1}}

\newcommand{\C}{{\mathcal{C}}}
\newcommand{\bF}{{\mathbb F}}
\newcommand{\tabincell}[2]{\begin{tabular}{@{}#1@{}}#2\end{tabular}}

\begin{document}

\title{A Structural Characterization of Cyclotomic Cosets with Applications to Affine-Invariant Codes and BCH Codes}

\author{ Xiongkun~Zheng,\,\, 
         Dabin~Zheng,\,\,
         Xiaoqiang~Wang,\,\,
         Mu~Yuan\,\,
\thanks{ X. Zheng, D. Zheng, X. Wang and M. Yuan are with Hubei Key Laboratory of Applied Mathematics, Faculty of Mathematics and Statistics, Hubei University. Dabin Zheng is also with Key Laboratory of Intelligent Sensing System and Security (Hubei University), Ministry of Education, Wuhan, 430062, China. E-mail: 202321104011243@stu.hubu.edu.cn; dzheng@hubu.edu.cn; waxiqq@163.com; yuanmu847566@outlook.com. The corresponding author is Dabin Zheng. The research was supported in part by the National Nature Science Foundation of China (NSFC) under Grants 62272148, 12301671.
}
}

\markboth{IEEE IT,~Vol.~, No.~, }%
{Shell \MakeLowercase{\textit{et al.}}:  A Structural Characterization of Cyclotomic Cosets with Applications to Affine-Invariant Codes and BCH Codes}

\maketitle

\begin{abstract}
Affine-invariant codes have attracted considerable attention due to their rich algebraic structure and strong theoretical properties. In this paper, we study a family of affine-invariant codes whose defining set consists of all descendants of elements in the cyclotomic coset of a single specified element. Our main contributions are as follows. First, we establish a new combinatorial result that determines exactly the size of such descendant sets, which is of independent interest in the study of cyclotomic cosets. Second, using this result, we derive explicit formulas for the dimensions of the corresponding affine-invariant codes and their associated cyclic codes, and we establish lower bounds on the minimum distances of their duals. In particular, under appropriate parameter choices, these codes yield narrow-sense primitive BCH codes and their extended counterparts. For the special class of narrow-sense primitive BCH codes with designed distance $\delta = (b+1)q^{m-t-1}$, where $1 \leq b \leq q-1$ and $0 \leq t \leq m-1$, we provide exact dimension formulas and an improved lower bound on the minimum distance. The results presented here extend and sharpen several previously known results, and provide refined tools for the parametric analysis of BCH codes and their duals.
\end{abstract}

\begin{IEEEkeywords}
Affine-invariant code,\ BCH code,\ dimension,\ cyclotomic coset, \ minimum distance.
\end{IEEEkeywords}

\section{Introduction}\label{sec:intro}

Affine-invariant codes, introduced by Kasami, Lin, and Peterson~\cite{Kasami1967} as a natural generalization of Reed--Muller codes, form an important class of linear codes characterized by their invariance under affine permutations of coordinate positions. Owing to their rich algebraic structure and strong symmetry properties, affine-invariant codes have played a significant role in coding theory, combinatorics, and cryptography. The seminal work of Kasami et al.~\cite{Kasami1967} established a necessary and sufficient condition, formulated in terms of the combinatorial structure of the defining set, for an extended cyclic code to be affine-invariant. Subsequently, Delsarte~\cite{Delsarte1970} demonstrated that every linear $q$-ary code invariant under a suitable permutation group is equivalent to a cyclic code, and that any nontrivial linear code of length $q^{rm}$ over $\mathbb{F}_q$ invariant under the general affine group is equivalent to an extended cyclic code. Further structural properties, including the group-code characterization and automorphism groups of affine-invariant codes and related extended BCH codes, were investigated in~\cite{Berger1996,BergerCharpin1996,BergerCharpin1999,Bernaletal2011}.

More recently, Ding and Tang~\cite{Ding2018} uncovered a remarkable connection between affine-invariant codes and combinatorial $2$-designs by showing that the supports of codewords with fixed weights naturally form $2$-designs. This important discovery has stimulated extensive research activity, leading to the construction of several infinite families of combinatorial designs derived from affine-invariant codes~\cite{DingLiXia2018,DingTangTonChev2020,DuWangFan2020}. These developments further highlight the fundamental importance of affine-invariant codes and deepen their connections with algebraic combinatorics. For a comprehensive account of affine-invariant codes and their applications, the reader is referred to~\cite{Hou2005,XuJiTaoHu2023,XuTaoHu2025} and the references therein.

Cyclotomic cosets constitute another fundamental tool in coding theory and finite field theory. The $q$-cyclotomic cosets modulo $n$ completely determine the irreducible factorization of $X^n - 1$ over $\mathbb{F}_q$, thereby playing a central role in the construction and analysis of cyclic codes. In addition, cyclotomic cosets arise naturally in the study of Gaussian sums, exponential sums, and related problems in number theory and combinatorics. A precise understanding of their structural properties, including coset representatives, leaders, sizes, and enumerations, is essential in many applications. In particular, for cyclic codes and affine-invariant codes, the cyclotomic coset structure of the defining set directly determines key code parameters such as the dimension and also affects bounds on the minimum distance. Consequently, the investigation of cyclotomic cosets is indispensable for both theoretical advances and practical constructions in coding theory.

Despite extensive research over several decades, determining the precise parameters of affine-invariant codes remains a challenging problem, especially for families defined via descendant structures of cyclotomic cosets. In an earlier work, Levy-dit-Vehel~\cite{Levy1995} studied a class of affine-invariant codes whose defining set consists of all descendants (see Definition~\ref{def:partialorder}) of elements belonging to the cyclotomic coset of a single specified element. By employing the generalized Roos bound, lower bounds on the minimum distances of the dual codes were obtained. However, due to the intricate combinatorial structure of the defining sets, explicit formulas for the code dimensions were not derived. More recently, Li et al.~\cite{ChengJuLi2025} investigated affine-invariant codes defined by descendants of elements from multiple cyclotomic cosets. While the dimensions of such codes can be determined relatively easily, obtaining tight minimum-distance bounds for these codes or their duals remains difficult. These limitations indicate that a deeper structural understanding of cyclotomic cosets and their descendant sets is crucial for advancing the theory of affine-invariant codes.

Motivated by these challenges, this paper investigates a class of affine-invariant codes $\C_T$ whose defining set $T$ consists of all descendants of elements in the cyclotomic coset of a single specified element. This framework generalizes the class considered in~\cite{Ding2015-1} and~\cite{Mann1962}. Our approach is based on a detailed combinatorial and structural analysis of cyclotomic cosets and their descendant sets. As a key step, we establish a new structural result that enables us to determine exactly the size of the defining set $T$. This structural result is of independent interest and provides new insights into the combinatorial properties of cyclotomic cosets.

Building on this structural result, we derive exact dimension formulas for the affine-invariant codes $\C_T$ and their associated cyclic codes $\C_T^*$. Furthermore, we establish lower bounds on the minimum distances of their dual codes. As an important application, we show that, for suitable parameter choices, the constructed affine-invariant codes and their cyclic counterparts include families of narrow-sense primitive BCH codes and extended narrow-sense primitive BCH codes. For these families, we obtain exact dimension formulas and improved lower bounds on the minimum distances. Our results not only strengthen and generalize but also unify several previously known results in~\cite{Ding2015-1,CLi2022,Levy1995,Mann1962,Wang2024}, and they provide a systematic framework for analyzing affine-invariant codes defined via cyclotomic cosets and descendant structures.

The remainder of this paper is organized as follows. Section~\ref{sec:preli} presents the necessary preliminaries. Section~\ref{sec:newresult} establishes the main structural result on cyclotomic cosets. Section~\ref{sec:dimensionofcodes} determines the dimensions of the proposed affine-invariant codes and their associated cyclic codes. Section~\ref{sec:distanceofcodes} derives lower bounds on the minimum distances of the dual codes. Finally, Section~\ref{sec:conclusion} concludes the paper and discusses possible directions for future research.

\section{Preliminaries}\label{sec:preli}

In this section, we introduce notation and recall several fundamental concepts concerning cyclotomic cosets, group algebras, cyclic codes, and affine-invariant codes that will be used throughout the paper.

\subsection{Basic notation and cyclotomic cosets}

Unless otherwise specified, the following notation will be used throughout the paper.

\begin{itemize}
\item $\bF_q$ denotes the finite field with $q$ elements, where $q$ is a prime power, and $\bF_{q^m}$ denotes its extension field of degree $m$.

\item For integers $k$ and $\ell$, the interval $[k,\ell]$ denotes the set $\{k,k+1,\dots,\ell\}$ if $k\le \ell$, and the empty set otherwise. For convenience, we write $[\ell]$ to denote $[0,\ell]$.

\item The binomial coefficient is defined by
\[
\binom{\ell}{k}=
\begin{cases}
\dfrac{\ell!}{k!(\ell-k)!}, & 0\le k\le \ell,\\
0, & \text{otherwise}.
\end{cases}
\]
We adopt the convention that $0^0=1$.
\end{itemize}

Let $m$ be a positive integer and let $n=q^m-1$. For an integer $s$ with $1\le s\le n$, the \textit{$q$-cyclotomic coset} of $s$ modulo $n$ is defined by
\[
\mathrm{cl}(s)=\{\, s, sq, sq^2,\dots,sq^{\ell_s-1}\,\} \pmod n,
\]
where $\ell_s$ is the smallest positive integer such that
\[
s \equiv sq^{\ell_s} \pmod n.
\]
The integer $\ell_s$ is called the \textit{size} of the cyclotomic coset $\mathrm{cl}(s)$, and the smallest element in $\mathrm{cl}(s)$ is called the \textit{coset leader}.

Every integer $s$ with $0\le s\le q^m-1$ admits a unique $q$-adic expansion
\[
s=\sum_{i=0}^{m-1}s_i q^i,\quad 0\le s_i\le q-1.
\]
This representation naturally corresponds to the sequence $\bar{s}=s_0 s_1 \cdots s_{m-1}$.
Conversely, each sequence of length $m$ over $\{0,1,\dots,q-1\}$ uniquely determines an integer in $[n]$.

Multiplication by $q^j$ modulo $q^m-1$ corresponds to a cyclic shift of the $q$-adic representation. More precisely, for $1\le j\le m$,
$sq^j \pmod{q^m-1}$ corresponds to the circular right shift
\[
s_{m-j} s_{m-j+1}\cdots s_{m-1}s_0 s_1\cdots s_{m-j-1}.
\]
Throughout the paper, we identify an integer with its $q$-adic sequence representation whenever convenient.

\begin{definition}
A \textit{cyclic subsequence} of length $k$ of a sequence
\[
\bar{x}=x_0x_1\cdots x_{m-1}
\]
is any subsequence consisting of $k$ consecutive symbols, where the sequence is viewed cyclically, i.e., the last symbol $x_{m-1}$ is followed by $x_0$.
\end{definition}

\begin{definition}\label{def:partialorder}
Define a partial order $\ll$ on $[n]$ as follows. For $u,v\in [n]$ with $q$-adic representations
\[
\bar{u}=u_0u_1\cdots u_{m-1}, \quad
\bar{v}=v_0v_1\cdots v_{m-1},
\]
we write $\bar{u} \ll \bar{v}$ if and only if
\[
u_i \le v_i,\quad \text{for all } 0\le i\le m-1.
\]
In this case, $u$ is called a \textit{descendant} of $v$, and $v$ is called an \textit{ascendant} of $u$. Clearly, $\bar{u}\ll \bar{v}$ implies $u\le v$.
\end{definition}

\subsection{Group algebras}

Let $G=\bF_{q^m}$ and let $G^*=\bF_{q^m}^*$ denote its multiplicative group. Consider the group algebras
$\mathcal{M}=\bF_q[G^*,\times]$ and $\mathcal{A}=\bF_q[G,+]$. Any element $\mathbf{x}$ in $\mathcal{M}$ or $\mathcal{A}$ can be written in the form
\[
\mathbf{x}=\sum_{g\in G^*} x_g X^g
\quad \text{or} \quad
\mathbf{x}=\sum_{g\in G} x_g X^g,
\]
respectively, where $x_g\in \bF_q$. Scalar multiplication and addition are defined componentwise:
\[
\lambda \sum_g x_g X^g=\sum_g \lambda x_g X^g, \quad
\sum_g x_g X^g+\sum_g y_g X^g=\sum_g (x_g+y_g)X^g.
\]
Multiplication in $\mathcal{M}$ is defined by
\[
\left(\sum_{g\in G^*} x_g X^g\right)
\left(\sum_{g\in G^*} y_g X^g\right)
=
\sum_{g\in G^*}
\left(
\sum_{\substack{\alpha\beta=g\\ \alpha,\beta\in G^*}}
x_\alpha y_\beta
\right)
X^g,
\]
while multiplication in $\mathcal{A}$ is defined by
\[
\left(\sum_{g\in G} x_g X^g\right)
\left(\sum_{g\in G} y_g X^g\right)
=
\sum_{g\in G}
\left(
\sum_{\substack{\alpha+\beta=g\\ \alpha,\beta\in G}}
x_\alpha y_\beta
\right)
X^g.
\]

Thus, $\mathcal{M}$ and $\mathcal{A}$ are $\bF_q$-vector spaces equipped with convolution-type multiplications induced by the multiplicative and additive group structures, respectively.

\subsection{Cyclic codes and extended cyclic codes}

An $[n,k,d]$ linear code $\C$ over $\bF_q$ is a $k$-dimensional subspace of $\bF_q^n$ with minimum distance $d$. The dual code of $\C$ is defined by
\[
\C^\perp=
\left\{
\mathbf{b}\in \bF_q^n
\mid
\sum_{i=0}^{n-1} b_i c_i=0,\,
\forall \mathbf{c}\in \C
\right\}.
\]
A linear code $\C$ is said to be \textit{cyclic} if $(c_0,c_1,\dots ,c_{n-1})\in \C$ implies $(c_{n-1},c_0,\dots ,c_{n-2})\in \C$.

Let $n=q^m-1$ and let $\alpha$ be a primitive element of $\bF_{q^m}$. By identifying a vector $(c_{\alpha^i})_{i=0}^{n-1}$ with the element
$\sum_{i=0}^{n-1} c_{\alpha^i} X^{\alpha^i}\in \mathcal{M}$, primitive cyclic codes of length $n$ correspond precisely to ideals of $\mathcal{M}$.
Define the $\bF_q$-linear map $\rho_s : \mathcal{M}\rightarrow\bF_{q^m}$ by
\[
\rho_s\left(\sum_{g\in G^*} x_g X^g\right)
=\sum_{g\in G^*} x_g g^s,\quad 0\le s<n.
\]
It follows immediately that $\rho_{qs}(\mathbf{x}) =\left(\rho_s(\mathbf{x})\right)^q$ for any $\mathbf{x}\in \mathcal{M}$.

\begin{definition}
Let $\C_T^*$ be a cyclic code of length $n$ over $\bF_q$. The \textit{defining set} of $\C_T^*$ is the largest subset $T^*\subseteq [n]$ such that
\[
\rho_s(\mathbf{x})=0,\quad \forall s\in T^*,\ \forall \mathbf{x}\in \C_T^*.
\]
\end{definition}

It is well known that $T^*$ is a union of $q$-cyclotomic cosets modulo $n$. The cyclic code with defining set $T^*$ can be written as
\[
\C_T^*
=
\left\{
\mathbf{c}\in \mathcal{M}
\mid
\rho_s(\mathbf{c})=0,\ \forall s\in T^*
\right\}.
\]
In particular, if
\[
T^*=\bigcup_{1\le s<\delta}\mathrm{cl}(s),
\]
then $\C_T^*$ is called the \textit{narrow-sense primitive BCH code} with design distance $\delta$, denoted by $\C_\delta$.

A cyclic code $\C_T^*$ can be extended by adding an overall parity-check coordinate. The extended code is defined by
\[
\C_T=\left\{
\left(-\sum_{g\in G^*} c_g\right)X^0
+
\sum_{g\in G^*} c_g X^g
\;\middle|\;
\sum_{g\in G^*} c_g X^g \in \C_T^*
\right\}.
\]
Its defining set is $T=T^*\cup\{0\}$. Moreover,
\[
\dim_{\bF_q} \C_T
=
\dim_{\bF_q} \C_T^*.
\]

\subsection{Affine-invariant codes}

For $u,v\in G$ with $u\ne 0$, define the affine transformation $\sigma_{u,v} : \mathcal{A}\rightarrow \mathcal{A}$ by
\[
\sigma_{u,v}
\left(
\sum_{g\in G} c_g X^g
\right)
=
\sum_{g\in G} c_g X^{ug+v}.
\]
This transformation permutes the coordinates of $\mathcal{A}$. A code $\C\subseteq \mathcal{A}$ is said to be \textit{affine-invariant} if
\[
\sigma_{u,v}(\mathbf{c})\in \C, \quad \forall \mathbf{c}\in \C,\ \forall u\ne 0,\ v\in G.
\]

The following fundamental characterization is due to Kasami et al.

\begin{lemma}\cite{Kasami1967}
Let $\C_T$ be an extended cyclic code with defining set $T$. Define
\[
\Delta(T)
=
\bigcup_{t\in T}
\{
s\in [n]\mid \bar{s}\ll \bar{t}
\}.
\]
Then $\C_T$ is affine-invariant if and only if
\[
\Delta(T)=T.
\]
\end{lemma}

The following lemma relates the dual distances of cyclic codes and their affine-invariant extensions.

\begin{lemma}\cite{Charpin1994}\label{lem:linkd}
If the extended cyclic code $\C_T$ is affine-invariant, then
\[
d_{\min}(\C_T^\perp)=d_{\min}((\C_T^*)^\perp),
\]
where $d_{\min}(\C_T^\perp)$ denotes the minimum distance of $\C_T^{\perp}$.
\end{lemma}

\section{A structural result on cyclotomic cosets}\label{sec:newresult}

It is well known that the size of a union of cyclotomic cosets plays a fundamental role in determining the dimension and minimum distance of cyclic codes, including BCH codes and affine-invariant codes. Therefore, obtaining an explicit characterization and enumeration of such unions is of both theoretical and practical importance.

Let $a,b,m,t$ be integers satisfying
\[
1 \le b \le a \le q-1, \quad m \ge 1, \quad 0 \le t \le m-1,
\]
and let $n=q^m-1$. Define the sequence
\[\bar{u}=\underbrace{a\cdots a}_{m-t-1} \, b \,\underbrace{0\cdots 0}_{t},
\]
and consider the set
\[T=\bigcup_{\bar{s}\ll \bar{u}} \mathrm{cl}(s),\]
where $\mathrm{cl}(s)$ denotes the $q$-cyclotomic coset of $s$ modulo $n$, namely,
\[
\mathrm{cl}(s)
=
\{s,sq,sq^2,\dots,sq^{\ell_s-1}\} \pmod n,
\]
and $\ell_s$ is the smallest positive integer such that $sq^{\ell_s}\equiv s \pmod n$. Equivalently, $\mathrm{cl}(s)$ is the orbit of $s$ under cyclic shifts of its $q$-adic expansion.
A special case of interest arises when $m=t+1$, and then $\bar{u}=b\underbrace{0\cdots0}_{m-1}$ (here $a$ may be chosen arbitrarily with $b\le a\le q-1$, since it does not affect the result).
The main objective of this section is to determine the exact size of the set $T$. To this end, we first establish a structural characterization of the elements of $T$ in terms of their $q$-adic expansions.

\begin{lemma}\label{lem:defT}
Let $\bar{u}=\underbrace{a\cdots a}_{m-t-1}\, b \,\underbrace{0\cdots0}_{t}$, $T=\bigcup_{\bar{s}\ll \bar{u}} \mathrm{cl}(s)$.
Then $T$ can be described as follows:
\[
T=\{0\}\cup
\left\{s\in [0,n]\;\middle|\;
\begin{array}{l}
\text{the $q$-adic expansion of $s$ contains a cyclic subsequence} \\[2pt]
\text{of the form}\,\, x\underbrace{0\cdots0}_{t},\,\, x\in [1,b],\,\, \text{or}\,\, y\underbrace{0\cdots0}_{t+1},\,\, y\in [b+1,a],\\[2pt]
\text{and all remaining digits are at most $a$.}
\end{array}
\right\}.
\]
\end{lemma}

\begin{IEEEproof}
By definition, $w \in T$ if and only if $w \in \mathrm{cl}(s)$ for some $\bar{s} \ll \bar{u}$.
Since $\mathrm{cl}(s)$ consists precisely of all cyclic shifts of $\bar{s}$, this condition is equivalent to $\bar{w} \ll \overline{u q^{\ell}}$ for some $\ell \in [0, m-1]$.

We now examine the cyclic shifts of $\bar{u}$. For $0\le \ell \le t$, the cyclic shift takes the form
\[\overline{u q^\ell}=
\underbrace{0\cdots0}_{\ell}\underbrace{a\cdots a}_{m-t-1}
\, b \,
\underbrace{0\cdots0}_{t-\ell}.
\]
For $t+1\le \ell \le m-1$, the cyclic shift takes the form
\[
\overline{u q^\ell}
=
\underbrace{a\cdots a}_{\ell-t-1}
\, b \,
\underbrace{0\cdots0}_{t}
\underbrace{a\cdots a}_{m-\ell}.
\]

Suppose first that $w\in T$ and $w\ne 0$. Then there exists some $\ell$ such that $\bar{w}\ll \overline{u q^\ell}$.
From the structure of $\overline{u q^\ell}$, it follows immediately that:
\begin{itemize}
\item all digits of $\bar{w}$ are at most $a$, and
\item $\bar{w}$ contains a cyclic subsequence of the form
$x\underbrace{0\cdots0}_{t}$, $x\in [0,b]$.
\end{itemize}

If $x\ne 0$, then this yields a subsequence of the form $x\underbrace{0\cdots0}_{t}$, $x\in [1,b]$, as required. If $x=0$, then $\bar{w}$ contains a cyclic subsequence consisting of at least $t+1$ consecutive zeros.
Since $w\ne 0$, there must exist a preceding nonzero digit. Tracing backward cyclically, we obtain a subsequence of the form $y\underbrace{0\cdots0}_{t+1}$, for some $y\in [1,a]$.

If $y\in [1,b]$, then the subsequence contains the required pattern $x\underbrace{0\cdots0}_{t},\,\,  x\in [1,b]$. If instead $y\in [b+1,a]$, then we obtain the required pattern
$y\underbrace{0\cdots0}_{t+1}$.

Conversely, suppose $\bar{w}$ contains one of the specified cyclic subsequences and all digits are at most $a$. Then it follows directly from the structure of the cyclic shifts of $\bar{u}$ that there exists some $\ell$ such that
$
\bar{w}\ll \overline{u q^\ell}.
$
Hence,
$
w\in T
$.

Finally, since $0\in \mathrm{cl}(0)$ and $\bar{0}\ll \bar{u}$, we have $0\in T$.

This completes the proof.
\end{IEEEproof}

To determine the cardinality of $T$, we partition it according to the number of occurrences of the admissible cyclic subsequences.

For integers $k,\ell\ge 0$, define
\begin{equation}\label{Bkl}
\mathcal{B}_{k,\ell}
=
\left\{
s\in [0,n]
\;\middle|\;
\begin{array}{l}
\text{the $q$-adic expansion of $s$ contains exactly}\\[4pt]
k \text{ cyclic subsequences of the form }
x\underbrace{0\cdots0}_{t}, \; x\in [1,b],\\[6pt]
\ell \text{ cyclic subsequences of the form }
y\underbrace{0\cdots0}_{t+1}, \; y\in [b+1,a],\\[6pt]
\text{and all other digits are at most } a.
\end{array}
\right\}.
\end{equation}

The admissible pairs $(k,\ell)$ are those satisfying the constraint
\begin{equation}\label{eq:omega}
\Omega
=
\left\{
(k,\ell)
\;\middle|\;
k(t+1)+\ell(t+2)\le m,\;
(k,\ell)\ne (0,0),\;
k,\ell\ge 0
\right\}.
\end{equation}

\begin{lemma}\label{lem:TBkl}
Let $T$ and $\mathcal{B}_{k,\ell}$ be defined as above. Then we have the following results.
\[T =  \bigsqcup_{(k,\ell)\in \Omega} \mathcal{B}_{k,\ell} \bigsqcup \{0\}, \,\, {\rm and } \,\, | T | = \sum_{(k,\ell) \in \Omega} |\mathcal{B}_{k,\ell}| + 1. \]
\end{lemma}
\begin{IEEEproof}
By Lemma~\ref{lem:defT}, if $ w \in T^*$, then the $q$-ary expansion of $w$ of length $m$ must contain at least one cyclic subsequence of the form either
\[ x\underbrace{0\cdots 0}_{t} \,\, \text{with}\ x \in [1,b], \,\, \text{or} \,\,
y\underbrace{0\cdots 0}_{t+1} \,\, \text{with}\ y \in [b+1,a],
\]
and all other digits are at most $a$. From the definition of $\mathcal{B}_{k,\ell}$, it follows directly that these sets are pairwise disjoint and
\[T = \bigsqcup_{(k,\ell) \in \Omega} \mathcal{B}_{k,\ell} \bigsqcup \{0\}.\]
Hence, $|T| = \sum_{(k,\ell) \in \Omega} |\mathcal{B}_{k,\ell}| + 1$.
\end{IEEEproof}

To determine the size of $T$, by Lemma~\ref{lem:TBkl}, we need to calculate the size of $\mathcal{B}_{k,\ell}$. But it is still difficult to
directly obtain the size of this set. To this end, we define the following set:
\begin{equation}\label{eq:Sxyz}
S _{k,\ell}(x,y,z) = \left\{ w \in \{x,y,z,0\}^m  \ \middle| \
\begin{array}{l}
\text{$w$ exactly contains $k$ cyclic-subsequences of  }\\
\text{the form $x\underbrace{0\dots 0}_{t}$, \,$\ell$ cyclic-subsequences of } \\
\text{the form $y\underbrace{0\dots 0}_{t+1}$, and all other digits are $z$}. \\
\end{array} \right\},
\end{equation}
where $k,\ell$ belong to $\Omega$ in (\ref{eq:omega}) and $x, y, z$ are viewed as symbols. Moreover, we define a matrix $\mathcal{A}_{k,\ell}$ as follows:
arrange all elements of set $S_{k,\ell}(x,y,z)$ into a row; then for each $w(x,y,z)\in S_{k,\ell}(x,y,z)$ in the first row vector, let
$x$ range over $[1, b]$, $y$ range over $[b+1, a]$ and $z$ range over $[0, a]$. In this way, we obtain a matrix $\mathcal{A}_{k,\ell}$ with $|S_{k,\ell}(x,y,z)|$ columns and
$b^k(a-b)^{\ell} (a+1)^{m - k(t+1)-\ell(t+2)}$ rows.  Thus, the total number of entries of $\mathcal{A}_{k,\ell}$ is
\begin{equation}\label{eq:defAkl}
 b^k (a-b)^{\ell} (a+1)^{m - k(t+1) - \ell(t+2)} \times |S_{k,\ell}(x,y,z)| \triangleq A_{k,\ell} .
\end{equation}

Next, we determine the size of $S_{k,\ell}(x,y,z)$.
\begin{lemma}\label{lem:sizeS}
Let $S_{k,\ell}(x,y,z)$ be defined as above. Then
\begin{equation}\label{eq:Skl}
\lvert S_{k,\ell}(x,y,z) \rvert = \frac{m}{m - kt - \ell(t+1)} \cdot \frac{[m - kt - \ell(t+1)]!}{k! \, \ell! \, [m - k(t+1) - \ell(t+2)]!}.
\end{equation}
\end{lemma}

\begin{IEEEproof}
Let \[
\bar{S}_{k,\ell}(x,y,z) = \left\{ w \in \{x\underbrace{0\dots0}_{t},\ y\underbrace{0\dots0}_{t+1},\ z\}^N \ \middle| \
\begin{array}{l}
\text{$w$ consists of $k$ copies of $x\underbrace{0\dots0}_{t}$,} \\
\text{$\ell$ copies of $y\underbrace{0\dots0}_{t+1}$, and the} \\
\text{remaining $N - k - \ell$ elements are $z$. }
\end{array} \right\}
\]
be the set of sequences of length $m$, where $N = m - kt - \ell(t+1)$. By combinatorial counting, we know that
\begin{equation}\label{eq:barSkl}
| \bar{S}_{k,\ell}(x,y,z) | = \binom{N}{k} \binom{N-k}{\ell} = \frac{[m - kt - \ell(t+1)]!}{k! \, \ell! \, [m - k(t+1) - \ell(t+2)]!}.
\end{equation}
Each element in $\bar{S}_{k,\ell}(x,y,z)$ has two representations:
\begin{itemize}
\item As $v_1 \cdots v_m$, where $v_i \in \{x, y, z, 0\}$, or
\item As $w_1 \| \cdots \| w_N$, where $w_i \in \{x\underbrace{0\dots0}_{t},\ y\underbrace{0\dots0}_{t+1},\ z\}$ and $\|$ denotes concatenation.
\end{itemize}

Let $\tau$ and $\sigma$ denote cyclic shift operators on sequences of the forms $v_1 \cdots v_m$ and $w_1 \| \cdots \| w_N$, respectively, i.e.,
\begin{align*}
\tau(v_1 \cdots v_m) &= v_m v_1 \cdots v_{m-1} \quad \text{for $v_i \in \{x, y, z, 0\}$}, \\
\sigma(w_1 \| \cdots \| w_N) &= w_N \| w_1 \| \cdots \| w_{N-1} \quad \text{for $w_i \in \{x\underbrace{0\dots0}_{t},\ y\underbrace{0\dots0}_{t+1},\ z\}$}.
\end{align*}
For a $w\in \bar{S}_{k,\ell}(x,y,z)$, we claim that $\sigma^d(w) = w$ for some positive integer $d$ if and only if $\tau^{d \cdot \frac{m}{N}}(w) = w$.
In fact, when $w$ is expressed as the form $w = w_1 \| \cdots \| w_N$ and $\sigma^d(w) = w$, we know that $w$ has the following form
\[
\underbrace{w_1 \| \cdots \| w_d \| \cdots \| w_1 \| \cdots \| w_d}_{\frac{N}{d} \text{ copies}},\,\, w_i \in \{x\underbrace{0\dots0}_{t},\ y\underbrace{0\dots0}_{t+1},\ z\}.
\]
It is easy to see that the sequence $w_1 \|\cdots \| w_d$ has length $\frac{md}{N}$, and then $\tau^{\frac{md}{N}}(w) = w$.
On the other hand, if $\tau^{d \cdot \frac{m}{N}}(w) = w$, then $w$ has the following form:
\[
\underbrace{v_1 \cdots v_{d \cdot \frac{m}{N}} \cdots v_1 \cdots v_{d \cdot \frac{m}{N}}}_{\frac{N}{d} \text{ copies}},\,\, v_i \in \{x, y, z, 0\}.
\]
By the definition of $\bar{S}_{k,\ell}(x,y,z)$, we know that
\[
v_1 \cdots v_{d \cdot \frac{m}{N}} = w_1 \| \cdots \| w_d,\,\,  w_i \in \{x\underbrace{0\dots0}_{t},\ y\underbrace{0\dots0}_{t+1},\ z\}.
\]
Thus, $\sigma^d(w) = w$.

From the definitions of $S_{k,\ell}(x,y,z)$ and $ \bar{S}_{k,\ell}(x,y,z)$, we have
\begin{equation*}\label{eq:relationSbarS}
S_{k,\ell}(x,y,z) = \left\{ \tau^i(w) \mid w \in \bar{S}_{k,\ell}(x,y,z),\ i \in [0, m-1] \right\}.
\end{equation*}
For any $w \in \bar{S}_{k,\ell}(x,y,z)$, by our claim we have that
\[ |\{\tau^i(w) \mid i \in [0, m-1]\} | = \frac{m}{N} \cdot | \{\sigma^i(w) \mid i \in [0, N-1]\} |. \]
Since $S_{k,\ell}(x,y,z)$ and $\bar{S}_{k,\ell}(x,y,z)$ are unions of equivalence classes under the cyclic shift operations of $\tau$ and $\sigma$, respectively, we deduce that
\begin{equation}\label{eq:rSklbarSkl}
 | S_{k,\ell}(x,y,z)| = \frac{m}{N} \cdot |\bar{S}_{k,\ell}(x,y,z)|.
\end{equation}
So, from (\ref{eq:barSkl}) and (\ref{eq:rSklbarSkl}) we obtain the equality (\ref{eq:Skl}). This completes the proof.
\end{IEEEproof}

By (\ref{eq:defAkl}) and Lemma~\ref{lem:sizeS} we get the size of the matrix $\mathcal{A}_{k,\ell}$:
\begin{equation}\label{eq:Akl}
A_{k,\ell}=b^k (a-b)^\ell (a+1)^{m - k(t+1) - \ell(t+2)}\frac{m}{m - kt - \ell(t+1)}\frac{[m - kt - \ell(t+1)]!}{k!\ell![m - k(t+1) - \ell(t+2)]!}.
\end{equation}

From the definition of $\mathcal{A}_{r,s}$, where $(r, s) \in \Omega$ in (\ref{eq:omega}), it is known that all entries of $\mathcal{A}_{r,s}$ come from elements of the set family
$\{\mathcal{B}_{k,\ell}\}_{(k,\ell)\in \Omega}$. We then examine the frequency with which each element $w$ of the set $\mathcal{B}_{k,\ell}$ appears in the matrix $\mathcal{A}_{r,s}$,
thereby establishing a quantitative relationship between the set families $\{A_{r,s}\}_{(r,s)\in \Omega}$ and $\{B_{k,\ell}\}_{(k,\ell)\in \Omega}$.

\begin{lemma}\label{lem:Ars}
Following the previous symbols and definitions, we have the following identity:
\begin{equation}\label{eq:ArsBkl}
A_{r,s} = \sum_{(k,\ell) \in \Omega} \binom{k}{r} \binom{\ell}{s} | \mathcal{B}_{k,\ell}|, \,\, (r, s) \in \Omega,
\end{equation}
where $A_{r,s}$ denotes the size of the matrix $\mathcal{A}_{r,s}$.
\end{lemma}
\begin{IEEEproof}
By the definition of $\mathcal{A}_{r,s}$, we know that its all entries come from the set families $\{\mathcal{B}_{k,l}\}_{(k,l) \in \Omega}$.
For any $w \in \mathcal{B}_{k,\ell}$, the $q$-ary expansion of $w$ of length $m$ contains exactly:
\begin{itemize}
\item $k$ cyclic subsequences of the form $x\underbrace{0\cdots0}_{t}$ with $x \in [1, b]$, and
\item $\ell$ cyclic subsequences of the form $y\underbrace{0\cdots0}_{t+1}$ with $y \in [b+1, a]$.
\end{itemize}
Without loss of generality, assume that
\[\bar{w} = x_1\underbrace{0\cdots0}_{t} \cdots x_k\underbrace{0\cdots0}_{t} y_1\underbrace{0\cdots0}_{t+1} \cdots y_\ell\underbrace{0\cdots0}_{t+1} \cdots \]
where $x_i \in [1, b]$ and $y_j \in [b+1, a]$.

By selecting $r$ out of the $k$ subsequences of the first type and $s$ out of the $\ell$ subsequences of the second type in $\bar{w}$, we obtain sequences of the form:
\[\cdots x_{i_1}\underbrace{0\cdots0}_{t} \cdots x_{i_r}\underbrace{0\cdots0}_{t} \cdots y_{j_1}\underbrace{0\cdots0}_{t+1} \cdots y_{j_s}\underbrace{0\cdots0}_{t+1} \cdots \]
Each such choice corresponds to an element in on column of $\mathcal{A}_{r,s}$, and the number of distinct choices is exactly $\binom{k}{r} \binom{\ell}{s}$.
So, each element $w \in \mathcal{B}_{k,\ell}$ appears exactly $\binom{k}{r} \binom{\ell}{s}$ times in the matrix $\mathcal{A}_{r,s}$. Hence, the equalities in (\ref{eq:ArsBkl})
are derived.
\end{IEEEproof}

To determine the size $|\mathcal{B}_{k, \ell}|$ of $\mathcal{B}_{k, \ell}$,  we only need to solve the linear equations of (\ref{eq:ArsBkl}) since
$A_{r,s}$ has been given in (\ref{eq:Akl}). The following lemma gives the result.

\begin{lemma}\label{lem:sizeBkl}
Let the symbols $\Omega$, $A_{r,s}$ and $|\mathcal{B}_{k, \ell}|$ be as defined above. Then $|\mathcal{B}_{k, \ell}|$ can be linearly expressed in terms of $\{ A_{r,s}\}_{(r,s) \in \Omega}$ as follows:
\begin{equation}\label{eq:Bkl}
|\mathcal{B}_{k, \ell}| = \sum_{(r,s) \in \Omega} (-1)^{r+s-k-\ell} \binom{r}{k} \binom{s}{\ell} A_{r,s}.
\end{equation}
\end{lemma}
\begin{IEEEproof}
Equation \eqref{eq:ArsBkl} is a two-parameter binomial-transform relation, and the inversion can be obtained from the standard binomial inversion formula. Assume
\begin{equation}\label{eq:Bkl-Ars}
|\mathcal{B}_{k, \ell}|  = \sum_{(r,s) \in \Omega} c(r,s,k,\ell) A_{r,s},
\qquad c(r,s,k,\ell)\in \mathbb{Q}.
\end{equation}
Substituting \eqref{eq:ArsBkl} into \eqref{eq:Bkl-Ars}, we obtain
\[
|\mathcal{B}_{k, \ell}| = \sum_{(r,s) \in \Omega} c(r,s,k,\ell)\sum_{(u,v) \in \Omega} \binom{u}{r} \binom{v}{s} |\mathcal{B}_{u,v}|.
\]
Interchanging the order of summation gives
\begin{equation}\label{eq:BklBkl}
|\mathcal{B}_{k,\ell}| = \sum_{(u,v) \in \Omega} \left( \sum_{(r,s) \in \Omega} c(r,s,k,\ell) \binom{u}{r} \binom{v}{s} \right) |\mathcal{B}_{u,v}|.
\end{equation}
Since this identity holds for all $(u,v)\in\Omega$, we must have
\begin{equation}\label{eq:coeff}
\sum_{(r,s) \in \Omega} c(r,s,k,\ell) \binom{u}{r} \binom{v}{s} = \delta_{u,k} \delta_{v,\ell}
\qquad \text{for all } (u,v) \in \Omega.
\end{equation}

Now, by the binomial inversion identity,
\[
\sum_{r=0}^{u} (-1)^{r-k} \binom{u}{r} \binom{r}{k} = \delta_{u,k},
\qquad
\sum_{s=0}^{v} (-1)^{s-\ell} \binom{v}{s} \binom{s}{\ell} = \delta_{v,\ell}.
\]
Multiplying these two identities, we obtain
\begin{equation}\label{eq:coff1}
\sum_{r=0}^{u} \sum_{s=0}^{v} (-1)^{r+s-k-\ell} \binom{r}{k} \binom{s}{\ell} \binom{u}{r} \binom{v}{s} = \delta_{u,k} \delta_{v,\ell}.
\end{equation}
Comparing \eqref{eq:coeff} and \eqref{eq:coff1}, we may take
\[
c(r,s,k,\ell) = (-1)^{r+s-k-\ell} \binom{r}{k} \binom{s}{\ell}.
\]
Substituting this into \eqref{eq:Bkl-Ars} yields \eqref{eq:Bkl}.
\end{IEEEproof}

By Lemma~\ref{lem:TBkl}, Lemma~\ref{lem:sizeBkl} and (\ref{eq:Akl}), we have the following main theorem.
\begin{theorem}\label{thm:newresult}
Let $\bar{u} =  \overbrace{a \cdots a b\underbrace{0\cdots0}_t}^{m} $, and
$T = \bigcup_{\bar{s} \ll \bar{u}} \mathrm{cl}(s)$. Then
\begin{equation*}
\left| T \right|= \sum_{(k,\ell) \in \Omega} \sum_{(r,s) \in \Omega} (-1)^{r+s-k-\ell} \binom{r}{k} \binom{s}{\ell} A_{r,s} + 1,
\end{equation*}
where
\begin{equation*}
\Omega = \left\{ (k,\ell) \ \middle| \ k(t+1) + \ell(t+2) \leq m,\ (k,\ell) \neq (0,0),\ 0 \leq k, \ell \right\},
\end{equation*}
\begin{equation*}
A_{r,s} = b^r (a-b)^s (a+1)^{m - r(t+1) - s(t+2)} \cdot
\frac{m}{m - rt - s(t+1)} \cdot
\frac{[m - rt - s(t+1)]!}{r! \, s! \, [m - r(t+1) - s(t+2)]!},
\end{equation*}
and in the special case $t=m-1$, the parameter $a$ may be chosen arbitrarily subject to $b\leq a\leq q-1$.
\end{theorem}

Next, we give an example to illustrate the proof process of Theorem~\ref{thm:newresult}.

\begin{example}
Let $q = 3$, $m=4$, $t=1$ and $\bar{u} = 2\, 2\, 1\, 0$. Let $T = \bigcup_{\bar{s} \ll \bar{u}} \mathrm{cl}(s)$.
In the following, we determine the size of $T$ by the steps in Theorem~\ref{thm:newresult}.

It is easy to see that
\[ \Omega = \left\{(k, \ell)\, |\, 2k+3\ell \leq 4, \,(k, \ell )\neq (0, 0) \right\} = \left\{ (1,0),\ (2,0),\ (0,1) \right\}. \]
and it has been shown that
\[|T| = |\mathcal{B}_{1,0}| + |\mathcal{B}_{2,0}| +|\mathcal{B}_{0,1}| + 1.\]

To investigate the size of the sets $\mathcal{B}_{k,\ell}$ for $(k,\ell)\in \Omega$, we define the following sets:
\[
S_{1,0} = \{x0zz,\ zx0z,\ zzx0,\ 0zzx\}, S_{2,0} = \{x0x0,\ 0x0x\}, S_{0,1} = \{y00z,\ zy00,\ 0zy0,\ 00zy\},
\]
where $x=1, y=2$ and $z\in [0,2]$ as required. Arrange all elements of $S_{k,\ell}$, $(k ,\ell)\in \Omega$ into the first row of a matrix $\mathcal{A}_{k,\ell}$, then let $z$ range over $[0, 2]$ and obtain the corresponding matrices as follows:
\[
\mathcal{A}_{1,0} =
\begin{bmatrix}
1000 & 0100 & 0010 & 0001 \\
1001 & 1100 & 0110 & 0011 \\
1002 & 2100 & 0210 & 0021 \\
\underline{1010} & \underline{0101} & \underline{1010} & \underline{0101} \\
1011 & 1101 & 1110 & 0111 \\
1012 & 2101 & 1210 & 0121 \\
1020 & 0102 & 2010 & 0201 \\
1021 & 1102 & 2110 & 0211 \\
1022 & 2102 & 2210 & 0221
\end{bmatrix},
\qquad
\begin{aligned}
&\mathcal{A}_{2,0} =
\begin{bmatrix}
\underline{1010} & \underline{0101}
\end{bmatrix},
\\[1.5ex]
&\mathcal{A}_{0,1} =
\begin{bmatrix}
\uwave{2000} & \uwave{0200} & \uwave{0020} & \uwave{0002} \\
\uwave{2001} & \uwave{1200} & \uwave{0120} & \uwave{0012} \\
\uwave{2002} & \uwave{2200} & \uwave{0220} & \uwave{0022}
\end{bmatrix}.
\end{aligned}
\]

All sequences in above matrices belong to $\mathcal{B}_{k,\ell}$ for $(k, \ell)\in \Omega$. It is easy to see that the underlined sequences belong to $\mathcal{B}_{2,0}$, wavy-underlined sequences
belong to $\mathcal{B}_{0,1}$ and other sequences belong to $\mathcal{B}_{1,0}$. Let $A_{k,\ell}$ and $B_{k,\ell}$ denotes the size of the matrix $\mathcal{A}_{k,\ell}$ and the set $\mathcal{B}_{k,\ell}$ for $(k ,\ell)\in \Omega$, respectively. By analyzing the frequency of each element of $\mathcal{B}_{k,\ell}$ appearing in the matrices $\mathcal{A}_{r,s}$, we get
\[
\begin{bmatrix}
A_{1,0} \\
A_{2,0} \\
A_{0,1}
\end{bmatrix}
=
\begin{bmatrix}
1 & 2 & 0 \\
0 & 1 & 0 \\
0 & 0 & 1
\end{bmatrix}
\begin{bmatrix}
B_{1,0} \\
B_{2,0} \\
B_{0,1}
\end{bmatrix}.
\]
It is known that $A_{1,0}=36$, $A_{2, 0} =2$ and $A_{0,1} =12$. Solving this system we have
\[ B_{1,0} = 32, \, \, B_{2,0} = 2,\,\, B_{0,1} = 12 .\]
Thus $|T| = 32 + 2 + 12 + 1 = 47$.

On the other hand, we can list all elements of $T$ as follows:
\[\begin{aligned}
T  =& \big\{  0, 1, 2, 3, 4, 5, 6, 7, 8, 9, 10, 11, 12, 13, 14, 15, 16, 17, 18, 19, 21, 24, 27, 28, 29, \\
  & 30, 31, 32, 33, 36, 37, 39, 42, 45, 46, 48,  51, 54, 55, 56, 57, 58, 59, 63, 64, 72, 73 \big\},
\end{aligned}
\]
and $T$ has exactly 47 elements. This example verifies the correctness of Theorem~\ref{thm:newresult}.
\end{example}

\section{The dimensions of a class of affine-invariant codes and related cyclic codes}\label{sec:dimensionofcodes}

In this section, we apply the result established in Section~\ref{sec:newresult} to determine the dimensions of a class of affine-invariant codes and the related cyclic codes.
To the best of our knowledge, existing work on the dimensions of cyclic codes has primarily focused on the BCH code family, with relatively few studies addressing the dimensions of cyclic codes with large dimensions.
Our theorems generalize several known results. To this end, we first recall the relevant definitions.

Let $1 \leq b \leq a \leq q-1$, $m \geq 1$, $0 \leq t \leq m-1$, and $n = q^m - 1$. Let
\[
\bar{u} = \overbrace{a\cdots a b \underbrace{0\cdots 0}_t}^{m},
\qquad
T = \bigcup_{\bar{s} \ll \bar{u}} \mathrm{cl}(s).
\]
It is known that the following linear code of length $q^m$ with defining set $T$,
\[
\C_T = \left\{  {\bf c} = \sum_{g\in \bF_{q^m}} c_g X^g \in \mathcal{A} \,\middle| \, \sum_{g\in \bF_{q^m}} c_g g^s = 0, \,\, \forall  s\in T \right\},
\]
is affine-invariant. This code can be regarded as the extension of the following cyclic code with defining set $T^*$,
\[
\C_T^* = \left\{  {\bf c} = \sum_{g\in \bF_{q^m}^*} c_g X^g \in \mathcal{M} \,\middle| \, \sum_{g\in\bF_{q^m}^*} c_g g^s = 0, \,\, \forall  s\in T^* \right\},
\]
where $T^* = T \setminus \{0\}$.

By Theorem~\ref{thm:newresult}, we directly obtain the following result.

\begin{theorem}\label{thm:dimCT}
Let the affine-invariant code $\C_T$ and the related cyclic code $\C_T^*$ be defined as above. Then
\[
{\rm dim}_{\mathbb{F}_q} (\C_T)
=
{\rm dim}_{\mathbb{F}_q} (\C_T^*)
=
q^m-1-
\sum_{(k,\ell) \in \Omega} \sum_{(r,s) \in \Omega}
(-1)^{r+s-k-\ell} \binom{r}{k} \binom{s}{\ell} A_{r,s},
\]
where
\[
\Omega = \left\{ (k,\ell) \,\, \middle| \,\, k(t+1) + \ell(t+2) \leq m,\ (k,\ell) \neq (0,0),\ 0 \leq k, \ell \right\},
\]
and
\[
A_{r,s}
=
b^r (a-b)^s (a+1)^{m - r(t+1) - s(t+2)}
\cdot
\frac{m}{m - rt - s(t+1)}
\cdot
\frac{[m - rt - s(t+1)]!}{r! \, s! \, [m - r(t+1) - s(t+2)]!}.
\]
\end{theorem}

Let
\[
\bar{u} = \overbrace{(q-1) \cdots (q-1) b \underbrace{0\cdots 0}_t}^{m}.
\]
Then
\[
u =(b+1)q^{m-t-1}-1.
\]
In this case, it is easy to see that for any integer $s\in [0, n]$,
\[
s\leq u \quad \text{if and only if} \quad \bar{s} \ll \bar{u}.
\]
Hence, $\C_T^*$ is a BCH code with design distance $\delta =u+1$, which we denote by $\C_{\delta}$. From Theorem~\ref{thm:dimCT}, we obtain the following corollary.

\begin{corollary}\label{cor:dimBCH}
Let $\C_\delta$ be a narrow-sense primitive BCH code of length $n = q^m - 1$ over $\mathbb{F}_q$ with design distance
\[
\delta = (b+1)q^{m-t-1},
\]
where $1 \leq b \leq q-1$ and $0 \leq t \leq m-1$. Then
\[
{\rm dim}_{\mathbb{F}_q}(\C_\delta)
=
q^m - 1 -
\sum_{(k,\ell) \in \Omega} \sum_{(r,s) \in \Omega}
(-1)^{r+s-k-\ell} \binom{r}{k} \binom{s}{\ell} A_{r,s},
\]
where
\[
\Omega = \left\{ (k,\ell) \,\, \middle| \,\,  k(t+1) + \ell(t+2) \leq m,\ (k,\ell) \neq (0,0),\ 0 \leq k, \ell \right\},
\]
and
\[
A_{r,s}
=
b^r (q-1-b)^s q^{m - r(t+1) - s(t+2)}
\cdot
\frac{m}{m - rt - s(t+1)}
\cdot
\frac{[m - rt - s(t+1)]!}{r! \, s! \, [m - r(t+1) - s(t+2)]!}.
\]
\end{corollary}

In conclusion, we have derived explicit dimension formulas for a class of affine-invariant codes and the related cyclic codes. To the best of our knowledge, this is the first time that the dimension of a cyclic code with a defining set of type $T^*$ has been determined explicitly. Ding et al.~in~\cite{Ding2015-1} gave a lower bound on the dimension of the BCH code in Corollary~\ref{cor:dimBCH}, whereas our result yields the exact dimension. Moreover, Corollary~\ref{cor:dimBCH} directly implies the related results in~\cite{Cherchem2020,Ding2015-1,Mann1962}. To the best of our knowledge, the known results on the dimensions of primitive BCH codes of length $q^m-1$ with design distance $\delta$ are summarized in Table~\ref{table:known_results}.

\begin{table}[h!]
\centering
{\small
\setlength{\tabcolsep}{0.8mm}
\caption{Known results on the dimension of $\mathcal{C}_{(q,m,\delta)}$}
\label{table:known_results}
 \begin{tabular}{cccc}  
\toprule
$m$ & $\delta$ & Dimension & Reference \\
\midrule
$m \geq 1$ & $\delta = q^t$ & Exact value & \cite{Mann1962} \\
\midrule
$m$ is odd & $2 \leq \delta \leq q^{(m+1)/2} + 1$ & Exact value & \cite{Yue1996} \\
\midrule
$m$ is even & $2 \leq \delta \leq 2q^{m/2} + 1$ & Exact value & \cite{Yue1996} \\
\midrule
\multicolumn{1}{c}{\tabincell{c}{$m \geq 1$}} & \multicolumn{1}{c}{\tabincell{c}{$\delta = (q - l_0)q^{m-l_1-1} - 1$\\with $0 \leq l_0 \leq q - 2$ and $0 \leq l_1 \leq m - 1$}} & \multicolumn{1}{c}{Lower Bound} & \multicolumn{1}{c}{\cite{Ding2015-1}} \\
\midrule
$m \geq 1$ & $\delta = q^t - 1$ & Exact value & \cite{Ding2015-1} \\
\midrule
$m \geq 1$ & \multicolumn{1}{c}{\tabincell{c}{$\delta = q^t + b$\\with $1 \leq b \leq \lfloor (q^t - 1)/q^{m-t}\rfloor + 1$}} & Exact value & \cite{Ding2015-2} \\
\midrule
$m \geq 1$ & $\delta = (q - 1)q^{m-1} - 1 - q^{\lfloor (m-1)/2\rfloor}$ & Exact value & \cite{Ding2017} \\
\midrule
$m \geq 4$ & $\delta = (q - 1)q^{m-1} - 1 - q^{\lfloor (m+1)/2\rfloor}$ & Exact value & \cite{Ding2017} \\
\midrule
$m$ is even & \multicolumn{1}{c}{\tabincell{c}{$\delta = k(q^{m/2} + 1)$\\with $1 \leq k \leq q - 1$}} & Exact value & \cite{Liu2017} \\
\midrule
$m \geq 4$ & $2 \leq \delta \leq q^{\lfloor (m+1)/2\rfloor + 1}$ & Exact value & \cite{Liu2017} \\
\midrule
$m = 2$ & $2 \leq \delta \leq q^m - 2$ & Exact value & \cite{Liu2017} \\
\midrule
$m > 1$ & \multicolumn{1}{c}{\tabincell{c}{$\delta = a(q^m - 1)/(q - 1)$ or $\delta = aq^{m-1} - 1$\\with $1 \leq a \leq q - 1$}} & Exact value & \cite{Cherchem2020} \\
\midrule
$m \geq 11$ & $\delta = (q - 1)q^{m-1} - 1 - q^{\lfloor (m+3)/2\rfloor}$ & Exact value & \cite{Moulou2023} \\
\midrule
$m > 1$ & \multicolumn{1}{c}{\tabincell{c}{$\delta = q^t + b$\\with $\lceil m/2 \rceil \leq t < m$ and \\$0 \leq b < q^{m-t} + \sum_{i=1}^{\lfloor t/s\rfloor - 1} q^{m+i(m-t)}$}} & Exact value & \cite{Gan2024} \\
\midrule
$m>1$ & \multicolumn{1}{c}{\tabincell{c}{$\delta = q^t - b$\\with $\lceil m/2 \rceil < t < m$ and \\$0 \leq b < (q - 1) \sum_{i=1}^{m-t} q^i$}} & Exact value & \cite{Gan2024} \\
\midrule
$m\geq 4$ & $2\leq \delta \leq q^{\lfloor \frac{2m-1}{3}\rfloor +1}$ & Exact value & \cite{RunZheng} \\
\midrule
$m \geq 1$ & \multicolumn{1}{c}{\tabincell{c}{$\delta = (b+1)q^{m-t-1}$\\with $1 \leq b \leq q - 1$ and $0 \leq t \leq m - 1$}} & Exact value & Corollary~\ref{cor:dimBCH} \\
\bottomrule
 \end{tabular}}
\end{table}

\section{The minimum distances of a class of affine-invariant codes and related cyclic codes}\label{sec:distanceofcodes}

In this section, we study the minimum distances of the dual codes of the affine-invariant codes and the related cyclic codes introduced in Section~\ref{sec:dimensionofcodes}. Our results generalize several previously known results. To this end, we give some definitions.

Let $1 \leq a \leq q-1$, $1 \leq b \leq q-1$, $m \geq 2$, $0 \leq t \leq m-1$, and $n = q^m - 1$. Let
\[
\bar{u} = \overbrace{a\cdots a b \underbrace{0\cdots 0}_t}^{m},
\qquad
T = \bigcup_{\bar{s} \ll \bar{u}} \mathrm{cl}(s).
\]
Let $\C_T$ and $\C_T^*$ be defined as in Section~\ref{sec:dimensionofcodes}, where $T$ and $T^* = T \setminus \{0\}$ are the defining sets of $\C_T$ and $\C_T^*$, respectively.
To investigate the minimum distance of the dual code of $\C_T$, we first recall the following known result.

\begin{lemma}[Theorem 1 in \cite{Charpin1994}]\label{lem:dualdefset}
Let $\C$ be an extended cyclic code in $\mathcal{A}$ with defining set $T$, and denote by $T^\perp$ the defining set of the dual of $\C$.
Then
\[
T^\perp=\{s\in [0, n]\,\,|\,\, n-s \notin T \}.
\]
\end{lemma}

Applying Lemma~\ref{lem:dualdefset}, we obtain the following characterization of the defining set of the dual code $\C_T^\perp$.
\begin{lemma}\label{lem:defTperp}
Let $u$, $T$, and $\C_T$ be as defined above.
Then the defining set of the dual code $\C_T^\perp$ is
\[
T^\perp = \left\{ w \in [n] \ \middle| \
\begin{array}{l}
\bar{w}\,\text{contains no cyclic subsequence of the form} \\
x_1\,\cdots x_{m-t-1} y\,\underbrace{q-1\,\cdots\,q-1}_{t},\ \text{with} \\
q-1-a \leq x_i \leq q-1 \ \text{and} \ q-1-b \leq y \leq q-1
\end{array} \right\}.
\]
\end{lemma}

\begin{IEEEproof}
It is known that
\[
T = \bigcup_{\bar{s} \ll \bar{u}} \mathrm{cl}(s), \quad \bar{u} = \overbrace{a\cdots a b \underbrace{0\cdots 0}_t}^{m}.
\]
Equivalently, by Lemma~\ref{lem:defT}, the set $T$ consists precisely of those $w\in [0, n]$ whose $q$-adic expansion contains a cyclic subsequence of the form
\[
x_1 \cdots x_{m-t-1} y \underbrace{0 \cdots 0}_{t},\quad x_i \in [0, a],\quad y \in [0, b].
\]
By Lemma~\ref{lem:dualdefset}, the defining set of $\C_T^\perp$ is
\[
T^\perp = \left\{  s \in [0, n]\, \mid\, n - s \notin T \right\}
       = \left\{  n - s \in [0, n] \, \mid\, s \notin T \right\}.
\]

Now note that passing from $s$ to $n-s=q^m-1-s$ amounts to replacing each digit of the $q$-adic expansion of $s$ by its complement with respect to $q-1$.
Therefore, the condition that $s$ contains a cyclic subsequence
$x_1 \cdots x_{m-t-1} y \underbrace{0\cdots 0}_t$, \,\, $x_i\in [0,a],\,\, y\in [0,b]$, is equivalent to the condition that $n-s$ contains a cyclic subsequence of the form
\[x_1'\cdots x_{m-t-1}' y' \underbrace{q-1 \cdots q-1}_t, \,\, q-1-a\le x_i' \le q-1,\,\, q-1-b\le y' \le q-1. \]
Hence, $n-s\notin T$ if and only if its $q$-adic expansion contains no such cyclic subsequence. This proves the stated characterization of $T^\perp$.
\end{IEEEproof}

By applying the Roos bound~\cite{Roos1983}, Levy-dit-Vehel obtained a lower bound on the minimum distance of affine-invariant codes via their defining sets, as follows.

\begin{lemma}[Theorem 4 in \cite{Levy1995}]\label{lem:Nbound}
Let $\C$ be an affine-invariant code of length $q^m$ over $\mathbb{F}_q$. Assume that its defining set $\bar{T}$ satisfies:
\begin{itemize}
\item $[0, v) \subseteq \bar{T}$ for some $v\in [n]$;
\item there exists an integer $z$ with $\gcd(z, q^m-1) = 1$, and a set $S \subseteq [n]$ with $0 \notin S$, such that for all $s \in S$,  $[sz, sz + v) \subseteq \bar{T}$;
\item $\max S - \min S - |S| + 1 < v$.
\end{itemize}
Then the minimum distance of $\C$ satisfies
\[
d_{\min}(\C) \geq v + |S| + 1.
\]
\end{lemma}

Next, we use Lemma~\ref{lem:Nbound} to derive lower bounds on the minimum distance of the dual code $\C_T^{\perp}$ of the affine-invariant code $\C_T$. To obtain a strong lower bound on the minimum distance of
$\C_T^{\perp}$, we first determine the largest possible integer $v\in [n]$ such that $[0,v) \subseteq T^{\perp}$ but $v\notin T^{\perp}$.
For the defining set $T^{\perp}$ given in Lemma~\ref{lem:defTperp}, we have the following result.

\begin{lemma}\label{lem:valueof}
Let $T^{\perp}$ be defined as in Lemma~\ref{lem:defTperp}. If an integer $v$ satisfies $[0,v) \subseteq T^{\perp}$ and $v\notin T^{\perp}$, then
\[
v =\begin{cases}
q^t-1, & a=b=q-1; \\[2pt]
q^{t+1}-1-b, & a=q-1,\ b\neq q-1,\ m\geq t+2; \\[2pt]
q^{t+1}-bq^t-1, & a=q-1,\ b\neq q-1,\ m= t+1; \\[2pt]
q^t\Bigl[(q-1-a)\sum_{i=0}^{m-t-2}q^i +1\Bigr]+q^{m-1}(q-1-b)-1,&a\neq q-1,\ b\geq a;\\[4pt]
q^{t+1}\Bigl[(q-1-a)\sum_{i=0}^{m-t-2}q^i +1\Bigr]-1-b,&a\neq q-1,\ b < a.
\end{cases}
\]
\end{lemma}

\begin{IEEEproof}
It is straightforward to verify that, in each of the above cases, $v\notin T^{\perp}$. Therefore, it remains to prove that $[0,v) \subseteq T^{\perp}$.
By the definition of $T^{\perp}$, this is equivalent to showing that for any $w \in [0, v)$, the $q$-adic expansion $\bar{w}$ contains no cyclic subsequence of the form
\[
x_1 \cdots x_{m-t-1} \, y \, \underbrace{q-1 \cdots q-1}_{t},
\]
where $q-1-a \leq x_i \leq q-1$ and $q-1-b \leq y \leq q-1$.

The structure of the set $T^{\perp}$ depends significantly on the parameters $a$, $b$, $m$, and $t$.
In particular, different parameter choices determine the maximal interval $[0,v)$ contained in $T^\perp$.
For this reason, we need to discuss $v$ according to the five cases listed in the statement.
In what follows, we present the proof only for the representative case $a \neq q-1$ and $b < a$; the remaining cases can be verified in a similar manner.

Assume now that $a \neq q-1$ and $b < a$.
In this case,
\[
v = q^{t+1}\Bigl[(q-1-a)\sum_{i=0}^{m-t-2} q^i + 1\Bigr] - 1 - b.
\]
Its $q$-adic expansion is
\[
\bar{v}=(q-1-b)\,\underbrace{q-1\,\cdots\,q-1}_{t}\, \underbrace{(q-1-a)\,\cdots\,(q-1-a)}_{m-t-1}.
\]
Hence $\bar{v}$ contains a cyclic subsequence of the forbidden form, and therefore $v \notin T^{\perp}$.

Now let $w \in [0, v)$ with $q$-adic expansion $\bar{w} = w_1 \cdots w_m$.
Since $w < v$, the first position at which $\bar{w}$ differs from $\bar{v}$ in lexicographic order must be strictly smaller than the corresponding digit of $\bar{v}$.
In particular, if a forbidden cyclic subsequence occurs in $\bar{w}$, then the corresponding digits force $w\ge v$, contradicting $w<v$. Indeed, if $\bar{w}$ contains no block of
$t$ consecutive digits equal to $q-1$, then clearly $w\in T^{\perp}$. So it suffices to consider the case where $\bar{w}$ does contain such a block.

There are then only two possible configurations for a forbidden cyclic subsequence:
\begin{itemize}
    \item $\bar{w} = \underbrace{q-1 \cdots q-1}_{t} w_{t+1} \cdots w_m$, where $q-1-a \leq w_{t+1}, \dots, w_{m-1} \leq q-1$ and $q-1-b \leq w_m \leq q-1$;
    \item $\bar{w} = w_1 \cdots w_{\ell} \underbrace{q-1 \cdots q-1}_{t} w_{\ell+t+1} \cdots w_m$ with $\ell \in [1, m-t-1]$, where
          $q-1-a \leq w_{\ell+t+1}, \dots, w_m, w_1, \dots, w_{\ell-1} \leq q-1$ and $q-1-b \leq w_{\ell} \leq q-1$.
\end{itemize}

In either configuration, the prescribed lower bounds on the digits imply that $w$ is at least the integer whose $q$-adic expansion is exactly
\[
(q-1-b)\,\underbrace{q-1\,\cdots\,q-1}_{t}\, \underbrace{(q-1-a)\,\cdots\,(q-1-a)}_{m-t-1},
\]
namely $v$. This contradicts the assumption $w<v$.
Therefore $\bar{w}$ cannot contain any forbidden cyclic subsequence, and hence $w\in T^{\perp}$.

Thus $[0,v)\subseteq T^\perp$ and $v\notin T^\perp$, which completes the proof.
\end{IEEEproof}

In the following, we construct a suitable set $S$ and an integer $z$ coprime to $q^m-1$, thereby deriving a lower bound on the minimum distance.

\begin{theorem}\label{thm:CTBound}
Let the affine-invariant code $\C_T$ and the related cyclic code $\C_T^*$ be defined as above. Then their dual codes have the same minimum distance~$d$, which satisfies the following lower bounds.

\noindent{\rm (1)} If $a = q-1$, $b \neq q-1$, then
\[
d \geq
\begin{cases}
q^{t+2}-qb-q,& m\geq 2t+5,\ t\geq 1;\\[2pt]
(q-b)(q^{t+1}-1-b),& m=2t+4,\ t\geq 1;\\[2pt]
2q^{t+1}-(b+1)q^t-1-b,& m=2t+3,\ t\geq 1;\\[2pt]
(q-1-b)(q-1)q^{m-t-3}+q^{t+1}-1-b,& m\in [t+3,2t+2],\ t\geq 1;\\[2pt]
q^{t+1}+q-2-2b,&m=t+2;\\[2pt]
(q-b)q^t,&m=t+1;\\[2pt]
\min\{ \lceil \frac{q}{b+1} \rceil ,q-b\}(q-1-b), & m\geq 3,\ t=0.
\end{cases}
\]

\noindent{\rm (2)} If $a = b = q-1$, $t\geq 1$, then
\[
d \geq
\begin{cases}
q^{t+1}-q+1, &  m \geq 2t+2; \\[2pt]
q^{m-t}-2q^{m-t-1}+q^{m-t-2}+q^t-1,&m \in [t+2,2t+1];\\[2pt]
q^t+q-2,&m=t+1.
\end{cases}
\]

\noindent{\rm (3)} If $a\neq q-1$, then
\[
d \geq
\begin{cases}
q^t\Bigl[(q-1-a)\sum_{i=0}^{m-t-2}q^i +1\Bigr]+q^{m-1}(q-1-b),&b\geq a;\\[4pt]
q^{t+1}\Bigl[(q-1-a)\sum_{i=0}^{m-t-2}q^i +1\Bigr]-b,&b < a.
\end{cases}
\]
\end{theorem}

\begin{IEEEproof}
By Lemma~\ref{lem:linkd}, we have
\[
d_{\min}\bigl( \C_T^\perp \bigr) = d_{\min}\bigl( (\C_{T}^*)^\perp \bigr).
\]
Hence it suffices to estimate the minimum distance of $\C_T^\perp$.

Lemma~\ref{lem:valueof} provides a value $v$ such that $[0,v)\subseteq T^\perp$ and $v\notin T^\perp$. According to Lemma~\ref{lem:Nbound}, a lower bound on $d$ can be obtained once we find an integer $z$ with $\gcd(z, q^m-1) = 1$ and a set $S \subseteq (0, q^m-1]$ such that, for every $s \in S$,
\begin{itemize}
    \item $[sz,\; sz + v) \subseteq T^\perp$,
    \item $\max S - \min S - |S| + 1 < v$.
\end{itemize}
Under these conditions, Lemma~\ref{lem:Nbound} yields $d \geq v + |S| + 1$.

The structure of the set $S$ depends strongly on the parameters $a$, $b$, $m$, and $t$.
For this reason, we proceed by cases.
Below we give a detailed proof only for one representative case; the remaining cases are handled similarly, and we record only the final conclusions.

\noindent \textbf{Case 1:} $a = q-1$, $b \neq q-1$, $m \geq 2t+5$, $t \geq 1$.

Let $z = q^{t+2}$ and
\[
S = \bigl\{ hq^{t+1} + k \mid (h,k) \neq (0,0), \ h \in [0, q-2], \ k \in [0, q^{t+1} - 2 - b] \bigr\}.
\]
We show that for all $s \in S$, one has $[sz, sz + v) \subseteq T^\perp$.
Equivalently, for all $w \in [0, v)$, we must prove that $sz + w \in T^\perp$.

For any $w \in [0, v)$ and $s \in S$, write
\[
\bar{w} = w_1 \cdots w_{t+1} \underbrace{0 \cdots 0}_{m-t-1},
\qquad
\bar{s} = s_1 s_2 \cdots s_{t+1} s_{t+2} \underbrace{0 \cdots 0}_{m-t-2}.
\]
Then
\[
\overline{sz + w} = w_1 \cdots w_{t+1} \, 0 \, s_1 s_2 \cdots s_{t+1} s_{t+2} \, 0 \, \underbrace{0 \cdots 0}_{m-2t-5}.
\]
We have to show that $\overline{sz + w}$ contains no cyclic subsequence of the form
\[
y \, \underbrace{q-1 \cdots q-1}_{t},
\]
with $q-1-b \leq y \leq q-1$.

By the definition of $S$, we have $s_{t+2} \leq q-2$.
If $\overline{sz + w}$ does not contain $t$ consecutive digits equal to $q-1$, then it certainly belongs to $T^{\perp}$.
Therefore, it suffices to consider the case where $\overline{sz + w}$ does contain $t$ consecutive $q-1$'s.
In that case, only the following four configurations are possible:
\begin{itemize}
    \item $\overline{sz + w} = \underbrace{q-1 \cdots q-1}_{t} w_{t+1}0 s_1s_2 \cdots s_{t+1} s_{t+2} 0 \underbrace{0 \cdots 0}_{m-2t-5}$;
    \item $\overline{sz + w} = w_1 \underbrace{q-1 \cdots q-1}_{t} 0 s_1s_2 \cdots s_{t+1} s_{t+2} 0 \underbrace{0 \cdots 0}_{m-2t-5}$;
    \item $\overline{sz + w} = w_1 \cdots w_{t+1} 0 \underbrace{q-1 \cdots q-1}_{t} s_{t+1} s_{t+2} 0 \underbrace{0 \cdots 0}_{m-2t-5}$;
    \item $\overline{sz + w} = w_1 \cdots w_{t+1} 0 s_1\underbrace{q-1 \cdots q-1}_{t} s_{t+2} 0 \underbrace{0 \cdots 0}_{m-2t-5}$.
\end{itemize}

In each of these four configurations, the digit immediately preceding the block of $t$ consecutive $q-1$'s is either one of the digits coming from $w$ or one of the digits coming from $s$.
Because $w<v=q^{t+1}-1-b$, the digits of $\bar w$ cannot produce a preceding digit in the interval $[q-1-b,q-1]$ together with such a block; similarly, by the definition of $S$, we have $s_{t+2}\le q-2$, so the block arising from the $s$-part cannot produce a forbidden pattern either. Therefore $\overline{sz+w}$ contains no cyclic subsequence of the forbidden form, and then $sz+w\in T^\perp$.

Thus $[sz,sz+v)\subseteq T^\perp$ for every $s\in S$.
A direct calculation shows that
\[
\max S - \min S - |S| + 1 < v.
\]
Hence all conditions of Lemma~\ref{lem:Nbound} are satisfied, and then
\[
d \geq v + |S| + 1 = q^{t+2}-qb-q.
\]

The remaining cases can be handled analogously. We list the resulting choices of $z$ and $S$, together with the corresponding lower bounds.

\noindent \textbf{Case 2:} $a = q-1$, $b \neq q-1$, $m = 2t+4$, $t \geq 1$.
Let $z = q^{t+2}$ and
\[
S = \bigl\{ hq^{t+1} + k \mid (h,k) \neq (0,0), \ h \in [0, q-2-b], \ k \in [0, q^{t+1} - 2 - b] \bigr\}.
\]
Then, $d \geq (q-b)(q^{t+1}-1-b)$.

\noindent \textbf{Case 3:} $a = q-1$, $b \neq q-1$, $m = 2t+3$, $t \geq 1$.
Let $z = q^{t+2}$ and
\[
S = \bigl\{ hq^{t} + k \mid (h,k) \neq (0,0), \ h \in [0, q-2-b], \ k \in [0, q^{t} - 1] \bigr\}.
\]
Then, $d \geq 2q^{t+1}-(b+1)q^t-1-b$.

\noindent \textbf{Case 4:} $a = q-1$, $b \neq q-1$, $m \in[t+3,2t+2]$, $t \geq 1$.
Let $z = q^{t+1}$ and
\[
S = \bigl\{ hq^{m-t-2} + kq+\ell \mid (h,k,\ell) \neq (0,0,0), \ h \in [0, q-2-b], \ k \in [0, q^{m-t-3} - 1], \ \ell \in[0,q-2] \bigr\}.
\]
Then, $d \geq (q-1-b)(q-1)q^{m-t-3}+q^{t+1}-1-b$.

\noindent \textbf{Case 5:} $a = q-1$, $b \neq q-1$, $m = t+2$.
Let $z = q^{t+1}$ and
\[
S = \bigl\{ h \mid  h \in [1, q-2-b] \bigr\}.
\]
Then, $d \geq q^{t+1}+q-2-2b$.

\noindent \textbf{Case 6:} $a = q-1$, $b \neq q-1$, $m = t+1$.
Let $S = \emptyset$. Then, $d \geq v+1=(q-b)q^t$.

\noindent \textbf{Case 7:} $a = q-1$, $b \neq q-1$, $m \geq 3$, $t=0$.
Let $z = q$, $\alpha=\min \left\{ q-2-b, \left\lceil \frac{q}{b+1} \right\rceil - 2\right\}$ and
\[
S = \bigl\{ hq+k \mid (h,k) \neq (0,0),\ h \in [0, \alpha],\ k\in[0,q-2-b] \bigr\}.
\]
Then, $ d \geq \min\left\{ \left\lceil \frac{q}{b+1} \right\rceil ,q-b\right\}(q-1-b)$.

\noindent \textbf{Case 8:} $a = b = q-1$, $m \geq 2t+2$, $t\geq 1$.
Let $z = q^{t+1}$ and
\[
S = \bigl\{ hq^t+k \mid (h,k) \neq (0,0),\ h \in [0, q-2],\ k\in[0,q^{t}-2] \bigr\}.
\]
Then, $d \geq q^{t+1}-q+1$.

\noindent \textbf{Case 9:} $a = b = q-1$, $m \in [t+2,2t+1]$, $t\geq 1$.
Let $z = q^t$ and
\[
S = \bigl\{ hq^{m-t-1}+qk+\ell \mid (h,k,\ell) \neq (0,0,0),\ h \in [0, q-2],\ k\in[0,q^{m-t-2}-1],\ \ell \in [0, q-2] \bigr\}.
\]
Then, $d \geq q^{m-t}-2q^{m-t-1}+q^{m-t-2}+q^t-1$.

\noindent \textbf{Case 10:} $a = b = q-1$, $m = t+1$, $t\geq 1$.
Let $z = q^t$ and
\[
S = \bigl\{ h \mid h \in [1, q-2] \bigr\}.
\]
Then, $d \geq q^t+q-2$.

\noindent \textbf{Case 11:} $a \neq q-1$.
Let $S = \emptyset$. Then,
\[
d \geq v+1 =
\begin{cases}
q^t\Bigl[(q-1-a)\sum_{i=0}^{m-t-2}q^i +1\Bigr]+q^{m-1}(q-1-b),&b\geq a;\\[4pt]
q^{t+1}\Bigl[(q-1-a)\sum_{i=0}^{m-t-2}q^i +1\Bigr]-b,&b < a.
\end{cases}
\]

Finally, note that we do not consider the case $a=b=q-1$ and $t=0$, since then $\C_T=\{\mathbf{0}\}$ is the zero code, and the problem is trivial.
This completes the proof.
\end{IEEEproof}

When
\[
\bar{u} = \overbrace{(q-1) \cdots (q-1) b \underbrace{0\cdots 0}_t}^{m},
\]
that is, when $\delta= (b+1)q^{m-t-1}$, it is easy to see that $\C_T^*$, with
\[
T^* = \left(\bigcup_{\bar{s} \ll \bar{u}} \mathrm{cl}(s)\right)\setminus \{0\},
\]
is the narrow-sense primitive BCH code of designed distance $\delta$, which we denote by $\C_{\delta}$.
Applying Theorem~\ref{thm:CTBound} yields the following result.

\begin{corollary}\label{cor:BCHBound}
Let $\C_\delta$ be a narrow-sense primitive BCH code of length $n = q^m - 1$ over $\mathbb{F}_q$ with designed distance
\[
\delta=(b+1)q^{m-t-1},
\]
where $1\leq b\leq q-1$ and $0\leq t\leq m-1$. Let $d$ be the minimum distance of $\C_{\delta}^\perp$. Then
\[
d \geq
\begin{cases}
q^{t+2}-qb-q,& b \neq q-1,\ m\geq 2t+5,\ t\geq 1;\\[2pt]
(q-b)(q^{t+1}-1-b),& b \neq q-1,\ m=2t+4,\ t\geq 1;\\[2pt]
2q^{t+1}-(b+1)q^t-1-b,& b \neq q-1,\ m=2t+3,\ t\geq 1;\\[2pt]
(q-1-b)(q-1)q^{m-t-3}+q^{t+1}-1-b,& b \neq q-1,\ m\in [t+3,2t+2],\ t\geq 1;\\[2pt]
q^{t+1}+q-2-2b,& b \neq q-1,\ m=t+2;\\[2pt]
(q-b)q^t,&b \neq q-1,\ m=t+1;\\[2pt]
\min\left\{ \left\lceil \frac{q}{b+1} \right\rceil ,q-b\right\}(q-1-b), & b \neq q-1,\ m\geq 3,\ t=0;\\[2pt]
q^{t+1}-q+1, & b=q-1,\ m \geq 2t+2; \\[2pt]
q^{m-t}-2q^{m-t-1}+q^{m-t-2}+q^t-1,& b=q-1,\ m \in [t+2,2t+1];\\[2pt]
q^t+q-2,& b=q-1,\ m=t+1.
\end{cases}
\]
\end{corollary}

\begin{table}[h!]
\centering
{\small
\setlength{\tabcolsep}{1.8mm}
\caption{Lower bounds on the minimum distance of $\mathcal{C}_{\delta}^\perp$ when $q=5$, $\delta=(b+1)q^{m-t-1}$, and $m=10$}
\label{table:LiWang}
 \begin{tabular}{ccccccc}
\toprule
$\quad t \quad$ & $b$ & $\delta$ & Corollary~\ref{cor:BCHBound} & Ref. \cite{CLi2022} & Ref. \cite{Levy1995} & Ref. \cite{Wang2024} \\
\midrule
\tabincell{c}{8} & \tabincell{c}{1\\2\\3\\4} & \tabincell{c}{10\\15\\20\\25} & \tabincell{c}{1953126\\1953124\\1953122\\390640} & \tabincell{c}{1953124\\1953123\\1953122\\390625} & \tabincell{c}{1953124\\1953123\\1953122\\390640*} & \tabincell{c}{1953125\\1953124\\1953122\\390626} \\
\midrule
\tabincell{c}{7} & \tabincell{c}{1\\2\\3\\4} & \tabincell{c}{50\\75\\100\\125} & \tabincell{c}{390635\\390630\\390625\\78204} & \tabincell{c}{390634\\390643\\390652\\78125} & \tabincell{c}{390639*\\390638*\\390637*\\78204*} & \tabincell{c}{390635\\390644\\390652\\78126} \\
\midrule
\tabincell{c}{6} & \tabincell{c}{1\\2\\3\\4} & \tabincell{c}{250\\375\\500\\625} & \tabincell{c}{78183\\78162\\78141\\16024} & \tabincell{c}{78124\\78123\\78122\\15625} & \tabincell{c}{78203*\\78202*\\78201*\\16024*} & \tabincell{c}{78125\\78124\\78122\\15626} \\
\midrule
\tabincell{c}{5} & \tabincell{c}{1\\2\\3\\4} & \tabincell{c}{1250\\1875\\2500\\3125} & \tabincell{c}{15923\\15822\\15721\\5124} & \tabincell{c}{15624\\15623\\15622\\3125} & \tabincell{c}{16023*\\16022*\\16021*\\5124*} & \tabincell{c}{15625\\15624\\15622\\3126} \\
\midrule
\tabincell{c}{4} & \tabincell{c}{1\\2\\3\\4} & \tabincell{c}{6250\\9375\\12500\\15625} & \tabincell{c}{4623\\4122\\3621\\3121} & \tabincell{c}{3124\\3123\\3122\\625} & \tabincell{c}{5123*\\5122*\\5121*\\3120} & \tabincell{c}{3125\\3124\\3122\\626} \\
\midrule
\tabincell{c}{3} & \tabincell{c}{1\\2\\3\\4} & \tabincell{c}{31250\\46875\\62500\\78125} & \tabincell{c}{2492\\1866\\1242\\621} & \tabincell{c}{624\\623\\622\\125} & \tabincell{c}{2492\\1866\\1242\\620} & \tabincell{c}{625\\624\\622\\126} \\
\midrule
\tabincell{c}{2} & \tabincell{c}{1\\2\\3\\4} & \tabincell{c}{156250\\234375\\312500\\390625} & \tabincell{c}{615\\610\\605\\121} & \tabincell{c}{124\\123\\122\\25} & \tabincell{c}{492\\366\\242\\120} & \tabincell{c}{125\\124\\122\\26} \\
\bottomrule
 \end{tabular}
 }
\end{table}

In conclusion, we have established lower bounds on the minimum distances of a family of affine-invariant codes and the related family of BCH codes. To evaluate our results, we compare the minimum distance of the dual code $\C_{\delta}^{\perp}$, where $\delta=(b+1)q^{m-t-1}$, with several known bounds; the comparison is presented in Table~\ref{table:LiWang}. The data show that, in comparison with the results of \cite{CLi2022} and \cite{Wang2024}, our lower bounds for this class of BCH codes are generally tighter, particularly when $t$ is small. Table~\ref{table:LiWang} also contains a comparison with the results claimed in Corollary~2 of \cite{Levy1995}. Our examination shows that some of those claims are incorrect, and the affected entries are marked with asterisks in the table.

%

\section{Summary and concluding remarks}\label{sec:conclusion}

This paper has advanced the study of affine-invariant codes and BCH codes by establishing new combinatorial insights into their underlying cyclotomic structure. Our main contributions can be summarized as follows:
\begin{itemize}
\item We proved a key combinatorial result (Theorem~\ref{thm:newresult}), which provides an exact closed-form expression for the size of the set consisting of all descendants of elements in a single $q$-cyclotomic coset under the partial order induced by $q$-adic expansions. This result is fundamental, since it directly enables a precise parametric analysis of the associated codes.
\item By applying this combinatorial formula, we derived explicit dimension formulas for two closely related code families: the affine-invariant code $\mathcal{C}_T$ and its cyclic counterpart $\mathcal{C}_{T}^*$ (Theorem~\ref{thm:dimCT}), as well as the narrow-sense primitive BCH codes $\mathcal{C}_\delta$ with designed distance $\delta = (b+1)q^{m-t-1}$, where $1 \leq b \leq q-1$ and $0 \leq t \leq m-1$ (Corollary~\ref{cor:dimBCH}).
\item We established improved lower bounds on the minimum distances of the dual codes of these code families (Theorem~\ref{thm:CTBound} and Corollary~\ref{cor:BCHBound}).
\end{itemize}

Our work not only deepens the structural understanding of affine-invariant codes and their duals, but also provides concrete and computable formulas that sharpen several existing bounds in the literature. The methods developed in this paper also provide a refined framework for the parametric analysis of BCH codes. Moreover, the combinatorial approach introduced here has the potential to be extended to more general families of affine-invariant codes and their corresponding cyclic codes, which will be the subject of future work.


\begin{thebibliography}{10}


\bibitem{Berger1996}
T.P. Berger,
Automorphism groups and permutation groups of affine-invariant codes,
Finite fields and applications: Proceedings of the third international conference, Glasgow,
Jul. 1995.

\bibitem{BergerCharpin1996}
T. P. Berger, P. Charpin,
The permutation group of affine-invariant extended cyclic codes,
IEEE Transactions on Information Theory,
vol. 42, no. 6, pp. 2194-2209,
Nov. 1996.

\bibitem{BergerCharpin1999}
T. P. Berger, P. Charpin,
The automorphism groups of BCH codes and of some affine-invariant codes over extension fields,
Designs, Codes and Cryptography,
vol. 18, pp. 29-53,
Dec. 1999.

\bibitem{Bernaletal2011}
J. J. Bernal, Á. del Río, J. J. Simón,
Group code structures of affine-invariant codes,
Journal of Algebra,
vol. 325, no. 1, pp. 269-281,
Jan. 2011.

\bibitem{Charpin1994}
P. Charpin, Françoise Levy-Dit-Vehel,
On self-dual affine-invariant codes,
Journal of Combinatorial Theory, Series A,
vol. 67, no. 2, pp. 223-244,
Aug. 1994.

\bibitem{Cherchem2020}
A. Cherchem, A. Jamous, H. Liu, Y. Maouche,
Some new results on dimension and Bose distance for various classes of BCH codes,
Finite Fields and Their Applications,
vol. 65, 101673,
Aug. 2020.

\bibitem{Delsarte1970}
P. Delsarte,
On cyclic codes that are invariant under the general linear group,
IEEE Transactions on Information Theory,
vol. 16, no. 6, pp. 760-769,
Nov. 1970.

\bibitem{Ding2015-1}
C. Ding,
Parameters of several classes of BCH codes,
IEEE Transactions on Information Theory,
vol. 61, no. 10, pp. 5322-5330,
Oct. 2015.

\bibitem{Ding2015-2}
C. Ding, X. Du, Z. Zhou,
The Bose and minimum distance of a class of BCH codes,
IEEE Transactions on Information Theory,
vol. 61, no. 5, pp. 2351-2356,
May 2015.

\bibitem{Ding2017}
C. Ding, C. Fan, Z. Zhou,
The dimension and minimum distance of two classes of primitive BCH codes,
Finite Fields and Their Applications,
vol. 45, pp. 237-263,
May 2017.

\bibitem{Ding2018}
C. Ding,
Designs from linear codes,
World Scientific Publishing Co., Inc., River Edge, NJ, USA,
Oct. 2018.

\bibitem{DingLiXia2018}
C. Ding, C. Li, Y. Xia,
Another generalisation of the binary Reed-Muller codes and its applications,
Finite Fields and Their Applications,
vol. 53, pp. 144-174,
Sep. 2018.

\bibitem{DingTangTonChev2020}
C. Ding, C. Tang, V. Tonchev,
Linear codes of 2-designs associated with subcodes of the ternary generalized Reed-Muller codes,
Designs, Codes and Cryptography,
vol. 88, pp. 625-641,
Dec. 2019.

\bibitem{DuWangFan2020}
X. Du, R. Wang, C. Fan,
Infinite families of 2-designs from a class of cyclic codes,
Journal of Combinatorial Designs,
vol. 28,
Nov. 2020.

\bibitem{Gan2024}
C. Gan, C. Li, H. Qian, X. Shi,
On Bose distance of a class of BCH codes with two types of designed distances,
Designs, Codes and Cryptography,
vol. 92, pp. 2031-2053,
Mar. 2024.

\bibitem{CLi2022}
B. Gong, C. Ding, C. Li,
The dual codes of several classes of BCH codes,
IEEE Transactions on Information Theory,
vol. 68, no. 2, pp. 953-964,
Feb. 2022.

\bibitem{Hou2005}
X. Hou,
Enumeration of certain affine invariant extended cyclic codes,
Journal of Combinatorial Theory,
vol. 110, no. 1, pp. 71-95,
Apr. 2005.

\bibitem{XuJiTaoHu2023}
C. Ji, Y. Xu, S. Hu,
Extended cyclic codes sandwiched between Reed-Muller codes,
2021 IEEE International Symposium on Information Theory (ISIT),
Melbourne, Australia, pp. 444-449,
Jul. 2021.

\bibitem{Kasami1967}
T. Kasami, S. Lin, W.W. Peterson,
Some results on cyclic codes which are invariant under the affine group and their applications,
Information and Control,
vol. 11, no. 5-6, pp. 475-496,
Nov. 1967.

\bibitem{ChengJuLi2025}
C. Li, C. Gan,
A class of affine-invariant codes and their related codes,
SIAM Journal on Discrete Mathematics,
vol. 39, no. 4, pp. 2067-2101,
Oct. 2025.

\bibitem{Levy1995}
F. Levy-dit-Vehel,
Bounds on the minimum distance of the duals of extended BCH codes over \(\mathbb{F}_p\),
Applicable Algebra in Engineering, Communication and Computing,
vol. 6, pp. 175-190,
May 1995.

\bibitem{Liu2017}
H. Liu, C. Ding, C. Li,
Dimensions of three types of BCH codes over GF(q),
Discrete Mathematics,
vol. 340, no. 8, pp. 1910-1927,
Aug. 2017.

\bibitem{Mann1962}
H. B. Mann,
On the number of information symbols in Bose-Chaudhuri codes,
Information and Control,
vol. 5, no. 2, pp. 153-162,
Jun. 1962.

\bibitem{Moulou2023}
E. M. Mouloua, M. Najmeddine,
On the parameters of a class of narrow sense primitive BCH codes,
Journal of Algebra Combinatorics Discrete Structures and Applications,
vol. 10, no. 3,
Aug. 2023.

\bibitem{Roos1983}
C. Roos,
A new lower bound for the minimum distance of a cyclic code,
IEEE Transactions on Information Theory,
vol. 29, no. 3, pp. 330-332,
May 1983.

\bibitem{Wang2024}
X. Wang, C. Xiao, D. Zheng,
The duals of narrow-sense BCH codes with length $\frac{q^m-1}{\lambda}$,
IEEE Transactions on Information Theory,
vol. 70, no. 11, pp. 7777-7789,
Nov. 2024.

\bibitem{XuTaoHu2025}
Y. Xu, R. Tao, S. Hu,
New classes of affine-invariant codes sandwiched between Reed-Muller codes,
Finite Fields and Their Applications,
vol. 92,
Dec. 2023.

\bibitem{Yue1996}
D. Yue, Z. Hu,
On the dimension and minimum distance of BCH codes over GF(q),
Journal of Electronics (China),
vol. 13, pp. 216-221,
Jul. 1996.

\bibitem{RunZheng}
R. Zheng, N.-S. Sze, Z. Huang,
The dimension and Bose distance of certain primitive BCH codes,
IEEE Transactions on Information Theory,
vol. 71, no. 10, pp. 7670-7687,
Oct. 2025.


\end{thebibliography}
\end{document}